\documentclass{aims}
\usepackage{amsmath}
  \usepackage{paralist}
  \usepackage{graphics} 
  \usepackage{epsfig} 
\usepackage{graphicx}  \usepackage{epstopdf}
 \usepackage[colorlinks=true]{hyperref}
\hypersetup{urlcolor=blue, citecolor=red}

\def\currentvolume{X}
 \def\currentissue{X}
  \def\currentyear{200X}
   \def\currentmonth{XX}
    \def\ppages{X--XX}
     \def\DOI{10.3934/xx.xx.xx.xx}

\newtheorem{theorem}{Theorem}[section]
\newtheorem{corollary}{Corollary}

\newtheorem{lemma}[theorem]{Lemma}
\newtheorem{proposition}{Proposition}

\theoremstyle{definition}

\newtheorem{example}[theorem]{Example}

\title[Cyclic codes over $\mathbb{Z}_4+u\mathbb{Z}_4$] 
      {Cyclic codes over $\mathbb{Z}_4+u\mathbb{Z}_4$}

\author[Rama Krishna Bandi and Maheshanand Bhaintwal]{}


 \email{bandi.ramakrishna@gmail.com}
 \email{mahesfma@iitr.ac.in}

\begin{document}
\maketitle

\centerline{\scshape Rama Krishna Bandi }
\medskip
{\footnotesize
 \centerline{ Department of Mathematics}
  \centerline{Indian Institute of Technology Roorkee}
   \centerline{Roorkee-247667, INDIA}
} 

\medskip

\centerline{\scshape Maheshanand Bhaintwal}
\medskip
{\footnotesize
 \centerline{ Department of Mathematics}
  \centerline{Indian Institute of Technology Roorkee}
   \centerline{Roorkee-247667, INDIA}
}

\bigskip

\date{Received: XXX/ Accepted: XXX}
\maketitle

 \centerline{(Communicated by the associate editor name)}

\begin{abstract}
In this paper, we have studied cyclic codes over the ring $R=\mathbb{Z}_4+u\mathbb{Z}_4$, $u^2=0$. We have considered cyclic codes of odd lengths.  A sufficient condition for a cyclic code over $R$ to be a $\mathbb{Z}_4$-free module is presented. We have provided the general form of the generators of a cyclic code over $R$ and determined a formula for the ranks of such codes. In this paper we have mainly focused on  principally generated cyclic codes of odd length over $R$. We have determined a necessary condition and a sufficient condition for cyclic codes of odd lengths over $R$ to be $R$-free.
\end{abstract}

\section{Introduction}
Cyclic codes are amongst the most studied algebraic codes. Their structure is well known over finite fields \cite{macwilliams}. Recently codes over rings have generated a lot of interest after a breakthrough paper by Hammons et al. \cite{hammons} showed that some well known binary non-linear codes are actually images of some linear codes over $\mathbb{Z}_4$ under the Gray map. Since then, cyclic codes have also been extensively studied over various finite rings. Their structure over finite chain rings is now well known \cite{norton}. They have also been studied over other rings such as $\mathbb{F}_2+u\mathbb{F}_2$, $u^2=0$, \cite{bonnecaze}; $\mathbb{F}_2+u\mathbb{F}_2+v\mathbb{F}_2+uv\mathbb{F}_2$, $u^2=v^2=0, uv=vu$, \cite{yildiz1}; and $\mathbb{F}_2+v\mathbb{F}_2$, $v^2=v$, \cite{zhu}.

Bonnecaze and Udaya \cite{bonnecaze} have studied cyclic codes over the ring $\mathbb{F}_2+u\mathbb{F}_2$, $u^2=0$, and provided their basic framework. The ring $\mathbb{F}_2+u\mathbb{F}_2$ is useful because it shares many properties of $\mathbb{Z}_4$, and since it has characteristic $2$, it also shares properties of the field $\mathbb{F}_4$. In most of these studies length of the cyclic code is relatively prime to the characteristic of the ring. A complete structure of cyclic codes over $\mathbb{Z}_4$ of odd length has been given in \cite{pless} and \cite{kumar}.

 In this paper, we have studied cyclic codes over the ring $R=\mathbb{Z}_4+u\mathbb{Z}_4$, $u^2=0$.  We have considered cyclic codes of odd lengths. Recently, Yildiz and Karadeniz \cite{yildiz} have studied linear codes over $R$.  A linear code $C$ over $R$ can be expressed as $C = C_1 + uC_2$, where $C_1, C_2$ are linear codes over $\mathbb{Z}_4$. As usual, a cyclic code of length $n$ over $R$ is an ideal of $R_n=\frac{R[x]}{\left\langle x^n-1\right\rangle}$. We have shown that a linear code $C = C_1 + uC_2$ of length $n$ over $R$ is a cyclic code if and only if $C_1, C_2$ are cyclic codes of length $n$ over $\mathbb{Z}_4$. We have determined a sufficient condition for a cyclic code of odd length over $R$ to be a $\mathbb{Z}_4$-free module. We have provided the general form of the generators of a cyclic code over $R$, from which we have determined a formula for the ranks of such codes. The ring $R_n$ is in general not a principal ideal ring, and so a cyclic code over $R$ is in general not principally generated. In this paper we have mainly focused on cyclic codes of odd length over $R$ which are principally generated. We have determined a necessary condition and a sufficient condition for principally generated cyclic codes of odd lengths over $R$ to be $R$-free.

The paper is organized as follows: In Section II, we present the preliminaries. 
In Section III, we have discussed the Galois extensions of $R$ and the ideal structure of these extensions. In Section IV, we have studied cyclic codes of odd length over $R$. The forms of the ranks and minimal spanning sets of these codes are presented. In Section V, we have mainly focused to principally generated cyclic codes of odd length over $R$ and determined a necessary condition and a sufficient condition for cyclic codes over $R$ to be $R$-free. We have also expressed principally generated cyclic codes in terms of the $n^{th}$ roots of unity.

\section{Preliminaries}
Throughout the paper, $R$ denotes the ring $\mathbb{Z}_4+u\mathbb{Z}_4 = \{a + ub ~|~ a, b \in \mathbb{Z}_4\}$ with $u^2=0$. $R$ can be viewed as the quotient ring $\mathbb{Z}_4[u]/\left\langle u^2 \right\rangle$. The units of $R$ are
\[1, 3, 1+u, 1+2u, 1+3u, 3+u, 3+2u, 3+3u~,\]
and the non-units are
\[0,2, u, 2u, 2+u, 2+2u, 3u, 2+3u~.\]
$R$ has six ideals in all: $\{0\}, \left\langle u\right\rangle = \{0, u, 2u, 3u\}, \left\langle 2\right\rangle = \{0, 2, 2u, 2+2u\}, \left\langle 2u\right\rangle = \{0, 2u\}, \left\langle 2+u\right\rangle = \{0, 2+u, 2u, 2+3u\}$ and $\left\langle 2, u\right\rangle = \{0, 2, 2u, 3u, 2+u, 2+2u, 2+3u\}$.

$R$ is a local ring of characteristic 4 with $\left\langle 2, u\right\rangle$ as its unique maximal ideal. A commutative ring $\mathcal{R}$ is called a \emph{chain ring} if its ideals form a chain under the relation of inclusion. From the ideals of $R$, we can see that they do not form a chain; for instance, the ideals $\left\langle u\right\rangle$ and $\left\langle 2\right\rangle$ are not comparable. Therefore, $R$ is a non-chain extension of $\mathbb{Z}_4$.  Also $R$ is not a principal ideal ring; for example, the ideal $\left\langle 2, u\right\rangle$ is not generated by any single element of $R$.

We denote the residue field $\frac{R}{\left\langle 2, u\right\rangle}$ of $R$ by $\overline{R}$. Since $\{0+\left\langle 2, u\right\rangle\} \cup \{1+\left\langle 2, u\right\rangle\} = R$, therefore $\overline{R} \cong \mathbb{F}_2$. The image of any element $a\in R$ under the projection map $\mu: R \rightarrow \overline{R}$ is denoted by $\overline{a}$. The map $\mu$ is extended to $R[x] \rightarrow \overline{R}[x]$ in the usual way. The image of an element $f(x) \in R[x]$ in $\overline{R}[x]$ under this projection is denoted by $\overline{f}(x)$. A polynomial $f(x) \in R[x]$ is called \emph{basic irreducible (primitive)} if $\overline{f(x)}$ is an irreducible (primitive) polynomial in $\overline{R}[x]$.  Basic irreducible polynomials over finite local rings play approximately the same role as irreducible polynomials play over finite fields.

A polynomial $f(x)$ over $R$ is called a \textit{regular polynomial} if it is not a zero divisor in $R[x]$, equivalently, $f(x)$ is regular if $\overline{f(x)} \neq 0$.  Two polynomials $f(x), g(x) \in R[x]$ are said to be \emph{coprime} if there exist $a(x), b(x) \in R[x]$ such that
\[a(x)f(x) + b(x)g(x) = 1~.\]

Now we recall the Hensel's Lemma and factorization of polynomials in $\mathbb{Z}_4[x]$. A polynomial $f(x)$ in $\mathbb{Z}_4[x]$ is said to be \emph{primary} if the principal ideal  $\left\langle f \right\rangle$ is primary, i.e., whenever $ab \in \left\langle f \right\rangle$, then either $a \in \left\langle f \right\rangle$ or $b^j \in \left\langle f \right\rangle$ for positive integer $j$.

\begin{theorem}[Hensel's Lemma \cite{pless}]\label{hensels} Let $f$ be a monic polynomial in $\mathbb{Z}_4[x]$ and assume that $f~~(\mbox{mod}~2)= g_1 g_2\cdots g_r$, where $g_1,g_2,\ldots,g_r$ are pairwise coprime monic polynomials over $\mathbb{F}_2$. Then there exist pairwise coprime monic polynomials $f_1,f_2,\ldots,f_r$ over $\mathbb{Z}_4$ such that $f=f_1f_2\cdots f_r$ in $\mathbb{Z}_4[x]$ and $f_i~~(\mbox{mod} ~2)=g_i$, $i=1,2,\ldots,r$.
\end{theorem}

A Gray map  $\phi : R^n \rightarrow \mathbb{Z}_4^{2n}$ is defined by (see \cite{yildiz})
\[(\overline{a} + u\overline{b}) \mapsto (\overline{b}, \overline{a} + \overline{b})~.\]
The  Lee weight is defined on $R$ by
\[w_L(a+ub) = w_L(b, a+b)~,\]
where $w_L(b, a+b)$ is the usual Lee weight of $(b, a+b)$ in $\mathbb{Z}_4^2$. This weight is then extended componentwise to $R^n$. The Lee weight of an element $x \in R^n$ is the sum of the Lee weights of the coordinates of $x$.

\begin{theorem}\cite{yildiz}\label{thm2.1}
The Gray map $\phi : R^n \rightarrow \mathbb{Z}_4^{2n}$ is a distance preserving linear isometry with respect to the Lee weights in $R^n$ and $\mathbb{Z}_4^2$.
\end{theorem}

A linear code $C$ of length $n$ over $R$ is an $R$-submodule of $R^n$. $C$ may not be an $R$-free module. We can express $R^n$ as $R^n = \mathbb{Z}_4^n + u\mathbb{Z}_4^n$, and so a linear code $C$ of length $n$ over $R$ can be expressed as $C = C_1 + uC_2$, where $C_1, C_2$ are linear codes of length $n$ over $\mathbb{Z}_4$.
The Euclidean inner product of any two elements $x=(x_1, x_2, \ldots, x_n)$ and $y=(y_1, y_2, \ldots, y_n)$ of $R^n$ is defined as $x\cdot y = x_1y_1 + x_2y_2 + \cdots + x_ny_n$, where the operation is performed in $R$. The dual of a linear code $C$ is defined as $C^\perp = \{y \in R^n ~|~ x\cdot y = 0 ~\forall x \in C\}$. It follows immediately that if $C = C_1 + uC_2$ is a linear code over $R$, then $C^\perp = C_1^\perp + uC_2^\perp$. We define the \emph{rank} of a code $C$ as the minimum number of generators for $C$ and the \emph{free rank} of $C$ is the rank of $C$ if $C$ is a free module over $R$. There are two other codes associated with $C$, namely Tor$(C)$ and Res$(C)$ and are defined as Tor$(C)=\{ b \in \frac{\mathbb{Z}_4[x]}{\left\langle x^n-1 \right\rangle} ~:~ ub \in C \}$ and Res$(C)=\{a \in \frac{\mathbb{Z}_4[x]}{\left\langle x^n-1 \right\rangle} ~:~ a+ub \in C ~\mbox{for some} ~b \in \frac{\mathbb{Z}_4[x]}{\left\langle x^n-1 \right\rangle} \}$.

\section{Galois extension of $R$}

Let $n$ be an odd integer. We first consider the factorization of $x^n-1$ over $R$, as it plays a vital role in the study of cyclic codes over $R$ of length $n$.

\begin{theorem}\label{primitive} Let $g(x) \in \mathbb{F}_2[x]$ be a monic irreducible (primitive) divisor of $x^{2^r-1}-1$. Then there exists a unique monic basic irreducible (primitive) polynomial $f(x)$ in $R[x]$ such that $\overline{f(x)} = g(x)$ and $f(x)~|~(x^{2^r-1}-1)$ in $R[x]$.
\end{theorem}

\begin{proof}
Let $x^{2^r-1}-1 = g(x)g^\prime(x)$ in $\mathbb{F}_2[x]$. By Hensel's lemma, there exist $f(x), f^\prime(x) \in \mathbb{Z}_4[x]$ such that $x^{2^r-1}-1 = f(x)f^\prime(x)$ in $\mathbb{Z}_4[x]$ and $f(x)~(\textrm{mod}~2) = g(x)$, $f^\prime(x)~(\textrm{mod}~2) = g^\prime(x)$.
Since $\mathbb{Z}_4$ is a subring of $R$,  $f(x) \in R[x]$. Also $\overline{f(x)} = f(x)~(\textrm{mod}~\left\langle2, u\right\rangle) = g(x)$ and $f(x)~|~(x^{2^r-1}-1)$ in $R[x]$.
 \end{proof}

We call the polynomial $f(x)$ in Theorem (\ref{primitive}) the \emph{Hensel lift} of $g(x)$ to $R$.

Since $n$ is odd, it follows from \cite[Theorem XIII.11]{mcdonald} that $x^n-1$ factorizes uniquely into pairwise coprime basic irreducible polynomials over $R$. Let
\[x^n-1= f_1f_2\cdots f_m \]
be such a factorization of $x^n-1$. Then it follows from the Chinese Remainder Theorem that
\[\frac{R[x]}{\left\langle x^n-1\right\rangle} = \oplus_{i=1}^m \frac{R[x]}{\left\langle f_i\right\rangle}~.\]
Therefore every ideal $I$ of $\frac{R[x]}{\left\langle x^n-1\right\rangle}$ can be expressed as $I = \oplus_{i=1}^m I_i$, where $I_i$ is an ideal of the ring $R[x]/\left\langle f_i\right\rangle$, $i=1, 2, \ldots, m$.

Let us recall the Galois extension of $\mathbb{Z}_4$. Let $h(x)$ be a monic basic irreducible polynomial of degree $r$ in $\mathbb{Z}_4[x]$. Then the Galois ring $GR(4,r)$ over $\mathbb{Z}_4$ is defined as the residue class ring $\frac{\mathbb{Z}_4[x]}{\left\langle h(x) \right\rangle}$. The ring $GR(4,r)$ is a local ring with unique maximl ideal $\left\langle 2 \right\rangle$ and the residue field $\mathbb{F}_{2^r}$.

Let $\mathcal{T}=\{0,1, \xi, \xi ^2, \ldots, \xi ^{2^r-2} \}$ be the \textit{Teichm\"{u}ller} representatives of $GR(4,r)$, where $\xi$ is a root of a basic primitive polynomial of degree $r$ in $\mathbb{Z}_4[x]$. Then each element $a$ of $GR(4,r)$ can be written as $a=a_0+2a_1$, where $a_0,a_1 \in \mathcal{T}$. This representation is called the $2$-adic representation of elements of $GR(4,r)$.

Now we consider the Galois extension of $R$. Let $f(x)$ be a basic irreducible polynomial of degree $r$ in $R[x]$. Then the Galois extension of $R$ is defined as the quotient ring $\frac{R[x]}{\left\langle f(x) \right\rangle}$ and is denoted by $GR(R,r)$. 
If $\alpha$ is a root of $f(x)$ then the elements of $GR(R,r)$ can  uniquely be written as $m_0 + m_1 \alpha +m_2 \alpha^2 + \cdots + m_{r-1} \alpha^{r-1}$, $m_i \in  R$, $i=0,1, \ldots, r-1$, i.e. $GR(R,r)$ is free module of rank $r$ over $R$ with a basis $\{1, \alpha, \alpha^2, \ldots, \alpha^{r-1}\}$ and $|GR(R,r)|=16^r$. From Theorem (\ref{local}), it follows that the ring $GR(R,r)$ is a local ring with unique maximal ideal $\left\langle \left\langle 2,~u \right\rangle + \left\langle f \right\rangle \right\rangle$ and the residue field $\mathbb{F}_{2^r}$. Furthermore, 
\[GR(R,r) \simeq \frac{GR(4,r)[u]}{\left\langle u^2 \right\rangle} \simeq GR(4,r) \oplus uGR(4,r)~, \]
where $GR(4,r)$ is the Galois ring of degree $r$ over $\mathbb{Z}_4$  and $u^2=0$.

Therefore, an element $x$ of $GR(R,r)$ can be represented as $x=a+ub$, where $a,b \in GR(4,r)$. Using the $2$-adic representation of $a=a_0+2a_1$, $b=a_2+2a_3$, $a_0,a_1,a_2,a_3 \in \mathcal{T}$, the element $x \in GR(R,r)$ can further  be represented as $x=a_0+2a_1+ua_2+2ua_3$.

\begin{lemma} A non-zero element $x= a_0 + 2 a_1 + u a_2  + 2u a_3$ of $GR(R,r)$ is unit if and only if $a_0$ is non-zero in $\mathcal{T}$.
\end{lemma}

\begin{proof}
Since $x^4=a_0^4$ for every non-zero $x$ in $GR(R,r)$, the result follows.
\end{proof}

Thus the group of units of $GR(R,r)$, denoted by $GR(R,r)^{*}$, is given by \[GR(R,r)^{*}= \{ a_0 + 2 a_1 + u a_2 + 2u a_3  ~~:~~ a_0, a_1,a_2,a_3 \in \mathcal{T}, a_0 \ne 0 \}.\]

\begin{theorem} The group of units $GR(R,r)^{*}$ is a direct product of two groups $G_C$ and $G_A$, i.e., $GR(R,r)^*=G_C \times G_A$, where $G_C$ is a cyclic group of order $2^r-1$ and $G_A$ is an abelian group of order $8^r$.
\end{theorem}
\begin{proof} Let $\xi$ be a primitive element of $GR(R, r)$ and $G_C = \mathcal{T}^\ast = \{1, \xi, \ldots, \xi^{2^r-2}\}$. Then $G_C$ is a multiplicative cyclic group of order $2^r$. Let $x= a_0 + 2 a_1 + u a_2 + 2u a_3 \in GR(R,r)^*$. Define a mapping $\Gamma : GR(R,r)^* \longrightarrow G_C$ such that $\Gamma(x)=a_0$. It can  easily be seen that for any $\alpha, x, y \in GR(R,r)^*$, $\Gamma(\alpha x+y)=\Gamma(\alpha) \Gamma(x) + \Gamma(y)$. $\Gamma$ is obviously a surjective map. Therefore $\frac{GR(R,r)^*}{ker \Gamma} \simeq G_C$, where ker $\Gamma=\{ 1+2a_1+ua_2+2ua_3 ~:~ a_1,a_2,a_3 \in \mathcal{T}  \}$. Denote Ker $\Gamma$ by $G_A$ . 
Then it can easily be seen that $GR(R,r) \simeq G_C \times G_A$. Moreover, $|GR(R,r)^*| = |G_c||G_A|=8^r(2^r-1)$.
 \end{proof}

The set of all zero divisors of $GR(R,r)$ is given by $\{ 2a_1+ua_2+2ua_3 ~~ :~~ a_1,a_2,a_3 \in \mathcal{T}   \}$, which is maximal ideal generated by $\left\langle 2, u \right\rangle$ in $GR(R,r)$. 

Now we consider the ideal structure of $GR(R,r)$. We first prove the following Lemma.

\begin{lemma}
Let $f(x), g(x) \in R[x]$. Then $f(x), ~g(x)$ are coprime if and only if their images $\overline{f}(x), ~\overline{g}(x)$ are coprime in $\overline{R}[x]$.
\end{lemma}

\begin{proof}
If $f(x),~ g(x)$ are coprime, then it is immediate that $\overline{f}(x)$ and $\overline{g}(x)$ are coprime. Now suppose that $\overline{f}(x),~ \overline{g}(x)$ are coprime. Then there exist $a(x), b(x) \in R[x]$ such that
\[\overline{a}(x)\overline{f}(x) +\overline{b}(x)\overline{g}(x) = 1~.\]
Then there exits $r(x), s(x) \in R[x]$ such that
\begin{equation}\label{eq1}
a(x)f(x)+b(x)g(x) = 1+2r(x)+us(x)~.\end{equation}
Multiplying (\ref{eq1}) by $2r(x)$ and by $us(x)$, we respectively get equations:
\begin{eqnarray}
2r(x)a(x)f(x)+2r(x)b(x)g(x) = 2r(x)+2ur(x)s(x)~. \label{eq2}\\
us(x)a(x)f(x)+us(x)b(x)g(x) = us(x)+2ur(x)s(x)~.\label{eq3}
\end{eqnarray}
On adding (\ref{eq2}) and (\ref{eq3}), we get
\begin{equation}
a(x)(2r(x)+us(x))f(x)+b(x)(2r(x)+us(x))g(x) = 2r(x) + us(x)~.\label{eq4}
\end{equation}
Putting the value of $2r(x) + us(x)$ in (\ref{eq1}), we get
\[a(x)(1-2r(x)-us(x))f(x)+b(x)(1-2r(x)-us(x))g(x) = 1~.\]
Therefore $f(x)$ and $g(x)$ are coprime.
 \end{proof}

Now we consider the ideals of $R[x]/\left\langle f\right\rangle$, where $f$ is a basic irreducible polynomial over $R$. 

\begin{theorem}\label{local}
Let $f \in R[x]$ be a basic irreducible polynomial. Then the ideals of $R[x]/\left\langle f \right\rangle$ are precisely, $\{0\}$, $\left\langle 1+\left\langle f\right\rangle\right\rangle$, $\left\langle 2+\left\langle f\right\rangle\right\rangle$, $\left\langle u+\left\langle f\right\rangle\right\rangle$, $\left\langle 2u+\left\langle f\right\rangle\right\rangle$, $\left\langle 2+u+\left\langle f\right\rangle\right\rangle$ and $\left\langle \left\langle 2, u\right\rangle + \left\langle f \right\rangle \right\rangle$.
\end{theorem}

\begin{proof}
Let $I$ be a non-zero ideal of $R[x]/\left\langle f \right\rangle$. Let $h + \left\langle f \right\rangle \in R[x]/\left\langle f \right\rangle$. Since $f$ is basic irreducible, $\overline{f}$ is irreducible in $\overline{R}[x]$. Therefore gcd$(\overline{f}, \overline{h}) = 1$ or $\overline{f}$. Let gcd$(\overline{f}, \overline{h}) = 1$. Then $f$ and $h$ are coprime in R[x], and hence there exist $\lambda_1, \lambda_2 \in R[x]$ such that
\[\lambda_1f +\lambda_2h = 1~.\]
From this follows that $\lambda_2h = 1 (\textrm{mod}~f)$. Thus $h$ is an invertible element of $R[x]/\left\langle f \right\rangle$ and so $I = \left\langle 1+\left\langle f \right\rangle\right\rangle = R[x]/\left\langle f \right\rangle$.

Now suppose that gcd$(\overline{f}, \overline{h}) = \overline{f}$. Then there exists polynomials $g, f_1, f_2 \in R[x]$ such that
\[h = fg + 2f_1 + uf_2 ~,\]
and gcd$(\overline{f}, \overline{f}_1) =1$ or gcd$(\overline{f}, \overline{f}_2) =1$. It follows that $h + \left\langle f \right\rangle \in \left\langle \left\langle 2, u\right\rangle+\left\langle f\right\rangle\right\rangle$. Therefore if $I \ne \left\langle 1+\left\langle f\right\rangle\right\rangle$, then $I \subseteq \left\langle \left\langle 2, u\right\rangle+\left\langle f\right\rangle\right\rangle$. The non-zero ideals contained in $\left\langle \left\langle 2, u\right\rangle+\left\langle f\right\rangle\right\rangle$ are $\left\langle 2+\left\langle f\right\rangle\right\rangle$, $\left\langle u+\left\langle f\right\rangle\right\rangle$, $\left\langle 2u+\left\langle f\right\rangle\right\rangle$, $\left\langle 2+u+\left\langle f\right\rangle\right\rangle$ and $\left\langle \left\langle 2, u\right\rangle+\left\langle f\right\rangle\right\rangle$ itself. The result follows.
 \end{proof}

The Galois group $Gal(GR(R,r))$ of $Gal(R,r)$ is a cyclic group of order $(2^r-1)$, which is generated by the \emph{Frobenius automorphism} $\sigma$ on $GR(R,r)$ defined as $\sigma(x)=a_0^2+2a_1^2+ua_2^2+2ua_3^2$, where $x=a_0+2a_1+ua_2+2ua_3 \in R$. The automorphism $\sigma$ fixes the ring $R$.

\begin{example}\label{galois}
Consider the basic irreducible polynomial $h(x)=x^4+3x^3+2x^2+1$, which is the Hensel lift to $R$ of the polynomial $x^4+x^3+1 \in \mathbb{F}_2[x]$. Let $\xi$ be a root of $h(x)$. Then
 \begin{eqnarray*}
 \xi^4=\xi^3+2\xi^2+3, & \xi^5=3\xi^3+2\xi^2+3\xi+3, & \xi^6=\xi^3+\xi^2+3\xi+1, \\
 \xi^7=2\xi^3+\xi^2+\xi+3,  &  \xi^8=3\xi^3+\xi^2+\xi, & \xi^9=3\xi^2+3, \\
 \xi^{10}=3\xi^3+3\xi, & \xi^{11}=3\xi^3+\xi^2+1, & \xi^{12}=2\xi^2+\xi+1, \\
 \xi^{13}=2\xi^3+\xi^2+\xi, & \xi^{14}=3\xi^3+\xi^2+2\xi, & \xi^{15}=1.
 \end{eqnarray*}

 Let $\mathcal{T}=\{ 0,1,\xi,\xi^2,\xi^3,\xi^3+2\xi^2+3, 3\xi^3+2\xi^2+3\xi+3, \xi^3+\xi^2+3\xi+1, 2\xi^3+\xi^2+\xi+3, 3\xi^3+\xi^2+\xi, 3\xi^2+3, 3\xi^3+3\xi, 3\xi^3+\xi^2+1, 2\xi^2+\xi+1, 2\xi^3+\xi^2+\xi, 3\xi^3+\xi^2+2\xi \}$. Then $GR(R,4)=\{ a_0+2a_1+ua_2+2ua_3~ :~ a_i \in \mathcal{T}, i=0,1,2,3 \}$ and $|GR(R,4)|=4^{16}$.
 \end{example}

\section{Cyclic codes of odd length over $\mathbb{Z}_4+u\mathbb{Z}_4$}

 We assume that $n$ is odd throughout this section. For a finite chain ring $\mathcal{R}$, it is well known that the ring $\frac{\mathcal{R}[x]}{\left\langle x^n-1\right\rangle}$ is a principal ideal ring \cite{norton}. However, in the present case the ring $R$ is not a chain ring and the situation is not as straightforward. In fact, the ring $R_n = \frac{R[x]}{\left\langle x^n-1\right\rangle}$ is not in general a principal ideal ring, as the next result shows. The result is a generalization of \cite[Lemma 2.4]{yildiz1}.

\begin{theorem}
The ring $R_n = \frac{R[x]}{\left\langle x^n-1\right\rangle}$ is not a principal ideal ring.
\end{theorem}

\begin{proof}
Consider the augmentation mapping $\gamma: \frac{R[x]}{\left\langle x^n-1\right\rangle} \rightarrow R$ defined by
\[ \gamma(a_0+a_1x+\ldots +a_{n-1}x^{n-1}) = a_0~+~a_1~+~\ldots ~+~a_{n-1}~.\]
This is a surjective ring homomorphism. Consider now the ideal $I = \left\langle 2, u\right\rangle$ of $R$, which we know is not a principal ideal. Let $J = \gamma^{-1}(I)$. It is well known that the inverse image under a homomorphism of an ideal is an ideal. So  $J$ is an ideal of $\frac{R[x]}{\left\langle x^n-1\right\rangle}$. Now if we assume $J$ to be a principal ideal, then its homomorphic image $I$ must be principal, a contradiction. Hence $J$ is not a principal ideal and $\frac{R[x]}{\left\langle x^n-1\right\rangle}$ is therefore not a principal ideal ring.
 \end{proof}

Therefore, a cyclic code of length $n$ over $R$ is in general not principally generated.

 Since $n$ is odd, the ring $\frac{\mathbb{Z}_4[x]}{\left\langle x^n-1\right\rangle}$ is a principal ideal ring. Therefore a cyclic code of length $n$ over $R$ is of the form $C = C_1 + uC_2 = \left\langle g_1\right\rangle + u\left\langle g_2\right\rangle$, where $g_1, g_2 \in \mathbb{Z}_4[x]$ are generator polynomials of the cyclic codes $C_1$, $C_2$, respectively.

Let $\tau$ be the standard cyclic shift operator on $R^n$. A linear code $C$ of length $n$ over $R$ is cyclic if $\tau(c) \in C$ whenever $c \in C$, i. e., if $(c_0, c_1, \ldots, c_{n-1}) \in C$, then $(c_{n-1}, c_0, c_1, \ldots, c_{n-2}) \in C$. As usual, in the polynomial representation, a cyclic code of length $n$ over $R$ is an ideal of $\frac{R[x]}{\left\langle x^n-1\right\rangle}$.

\begin{theorem}
 Let $x^n-1= f_1f_2\cdots f_m$, where  $f_i$, $i=1,2,\ldots,m$ are basic irreducible pairwise coprime polynomials in $R[x]$. Then any ideal in $R_n$ is the sum of the ideals of $R[x]/\left\langle f_i\right\rangle$, $i=1, 2, \ldots, m$.
\end{theorem}

\begin{proof}
It follows from the  Chinese Remainder Theorem.
 \end{proof}

\begin{corollary}
The number cyclic codes over $R$ is $7^m$.
\end{corollary}
\begin{proof}
Each ideal of $R_n$ is a direct sum of the ideals of $R[x]/\left\langle f_i\right\rangle$, $i=1, 2. \ldots, m$. From Theorem (\ref{local}) and for each $i$, $R[x]/\left\langle f_i\right\rangle$ has $7$ ideals. The result follows.
 \end{proof}
\begin{theorem}\label{thm2.2}
A linear code $C = C_1 + uC_2$ of length $n$ over $R$ is cyclic if and only if $C_1$, $C_2$ are cyclic codes of length $n$ over $\mathbb{Z}_4$.
\end{theorem}

\begin{proof}
Let $c_1+uc_2 \in C$, where $c_1 \in C_1$ and $c_2 \in C_2$. Then $\tau(c_1+uc_2) = \tau(c_1) +u\tau(c_2) \in C$, since $C$ is cyclic and $\tau$ is a linear map. So, $\tau(c_1) \in C_1$ and $\tau(c_2) \in C_2$. Therefore $C_1, C_2$ are cyclic codes. Conversely if $C_1$, $C_2$ are cyclic codes, then for any $c_1+uc_2 \in C$, where $c_1 \in C_1$ and $c_2 \in C_2$, we have $\tau(c_1) \in C_1$ and $\tau(c_2) \in C_2$, and so, $\tau(c_1+uc_2) = \tau(c_1) +u\tau(c_2) \in C$. Hence $C$ is cyclic.
 \end{proof}

The following result gives a sufficient condition for a cyclic code $C$ over $R$ to be a free $\mathbb{Z}_4$-code.

\begin{theorem}\label{freeZ4}
Let $C = C_1 + uC_2$ be a cyclic code of length $n$ over $R$. If $C_1$, $C_2$ are free codes over $\mathbb{Z}_4$, then $C$ is  a free $\mathbb{Z}_4$-module.
\end{theorem}

\begin{proof}
Suppose that $C_1$, $C_2$ are $\mathbb{Z}_4$-free codes of ranks $k_1$, $k_2$, respectively. Let $\{c_{11}, c_{12}, \ldots, c_{1k_1}\}$ and $\{c_{21}, c_{22}, \ldots, c_{2k_2}\}$ be $\mathbb{Z}_4$-bases of $C_1$ and $C_2$, respectively. Then the set $\{c_{11}, c_{12}, \ldots, c_{1k_1}, uc_{21}, uc_{22}, \ldots,$ $uc_{2k_2}\}$ spans $C$, as every element of $C$ can be expressed as a linear combination of elements of this set. Now suppose there exist scalars $a_i, b_j \in \mathbb{Z}_4$ such that
\[ \sum_{i=1}^{k_1} a_ic_{1i} + u\sum_{j=1}^{k_2} b_jc_{2j} = 0~.\]
Then $\sum_{i=1}^{k_1} a_ic_{1i} =0$ and $\sum_{j=1}^{k_2} b_jc_{2j} = 0$. Since the elements $c_{11}, c_{12}, \ldots, c_{1k_1}$ are independent and so are the elements $c_{21}, c_{22}, \ldots, c_{2k_2}$, therefore $a_i=0$ and $b_j=0$ for all $i$ and $j$. Hence $C$ is a $\mathbb{Z}_4$-free module.
 \end{proof}

The converse of the above Theorem is not true in general, i. e., if a cyclic code $C = C_1 + uC_2$ is a free $\mathbb{Z}_4$-module of length $n$ over $R$, then $C_1$ or $C_2$ may not be a free code of length $n$ over $\mathbb{Z}_4$ (see example \ref{ex2}). However, if $C$ is an $R$-free module (code) of length $n$ over $R$ then $C_1$ must be a  free code of length $n$ over $\mathbb{Z}_4$ (see Theorem \ref{R free to z4 free}).

\begin{example}\label{ex1}
The polynomial $x^7-1$ factorizes into irreducible polynomials over $\mathbb{F}_2$ as $x^7-1 = (x-1)(x^3+x+1)(x^3+x^2+1)$. The Hensel lifts of $x^3+x+1$ and $x^3+x^2+1$ to $\mathbb{Z}_4 $ are $x^3+2x^2+x-1$ and $x^3-x^2-2x-1$, respectively. Therefore $x^3+2x^2+x-1$ and $x^3-x^2-2x-1$ are divisors of $x^7-1$ over $\mathbb{Z}_4$. Define $C = \left\langle x^3+2x^2+x-1\right\rangle + u \left\langle x^3-x^2-2x-1 \right\rangle$. Then $C$ is a cyclic code of length $7$ over $R$, which is also a free $\mathbb{Z}_4$-module.
\end{example}

\begin{example}\label{ex2} Let $C=C_1+uC_2$ be a free $\mathbb{Z}_4$-cyclic code of length $5$ over $R$ generated by $g(x)=u+2x+ux^2$. Then $C_1$ is a  cyclic code of length $5$ over $\mathbb{Z}_4$ generated by $g(x)~ (\mbox{mod}~ u)=2x$ which is not $\mathbb{Z}_4$- free.
\end{example}

Now we consider the general form of the generators of cyclic codes over $R$.

Define $\psi: R \rightarrow \mathbb{Z}_4$ such that $\psi(a+bu)=a~ (\mbox{mod}~ u)$. It can easily be seen that $\psi$ is a ring homomorphisms with ker $\psi$ $= \left\langle u \right\rangle$ $=u\mathbb{Z}_4$. Extended $\psi$ to the homomorphism $\phi: \frac{R[x]}{\left\langle x^n-1\right\rangle} \rightarrow \frac{\mathbb{Z}_4[x]}{\left\langle x^n-1\right\rangle}$ such that $\phi(a_0+a_1x+a_2x^2+ \ldots + a_{n-1}x^{n-1})=\psi(a_0)+\psi(a_1)x+\psi(a_2)x^2+ \ldots + \psi(a_{n-1})x^{n-1}$. Let $C$ be a cyclic code of length $n$ over $R$. Restrict $\phi$ to $C$ and define
\[J=\{ h(x) \in \frac{\mathbb{Z}_4[x]}{\left\langle x^n-1\right\rangle} ~:~ uh(x) \in \mbox{ker}~\phi\}~.\]
 Clearly $J$ is an ideal of $\frac{\mathbb{Z}_4[x]}{\left\langle x^n-1\right\rangle}$. So $J$ is a cyclic code over $\mathbb{Z}_4$ and $J=\left\langle a(x) \right\rangle$ for some $a(x) \in \mathbb{Z}_4[x]$. Therefore ker $\phi$ $= \left\langle ua(x) \right\rangle $. Similarly, the image of $C$ under $\phi$ is an ideal of $\frac{\mathbb{Z}_4[x]}{\left\langle x^n-1\right\rangle}$ and $\phi(C)=\left\langle g(x) \right\rangle$ for some $g(x) \in \mathbb{Z}_4[x]$. Hence $C=\left\langle g(x)+up(x),~ua(x) \right\rangle$ for some $p(x) \in \mathbb{Z}_4[x]$. Since $ug(x)= u(g(x)+up(x)) \in C$ and $\phi(ug(x))=0$, so $a(x) \mid g(x)$. Thus a cyclic code $C$ over $R$ has the form
 \[C = \left\langle g(x)+up(x), ua(x)\right\rangle~,\]
 where $g(x), p(x), a(x) \in \mathbb{Z}_4[x]$ and $a(x)\mid g(x)$. In particular if $a(x)=g(x)$, we have the following result.

\begin{theorem}
Let $n$ be an odd integer and $C$ be a cyclic code of length $n$ over $R$ such that $C=\left\langle g(x)+up(x), ~ug(x) \right\rangle$. Then $C=\left\langle g(x)+up(x) \right\rangle$.
\end{theorem}

\begin{proof}
Clearly $\left\langle g(x)+up(x) \right\rangle \subseteq C$. Since $u(g(x)+up(x))=ug(x)$ and $g(x)=a(x)$, $C \subseteq \left\langle g(x)+up(x) \right\rangle$. Hence $C=\left\langle g(x)+up(x) \right\rangle$.
\end{proof}

It may be noted here that unlike in the case of finite fields, a generator polynomial of ker $\phi$ or $\phi(C)$ may not necessarily divide $x^n-1$.
The proof of the following result is straightforward, as the result is well known for codes over finite fields.
\begin{theorem}\label{thm4.9}  Let $C$ be a cyclic code of length $n$ over $R$. If $C=\left\langle g(x)+up(x), ~ua(x) \right\rangle$ and $deg~g(x)=k_1$ and $deg~a(x)=k_2$, then $C$ has rank $2n-k_1-k_2$ and a minimal spanning set $A=\{ (g(x)+up(x)), x(g(x)+up(x)), x^2(g(x)+up(x)), \cdots, x^{n-k_1-1}(g(x)+up(x)),~ua(x), xua(x),x^2ua(x), \cdots, x^{n-k_2-1}ua(x) \}$.
\end{theorem}

In Theorem (\ref{thm4.9}), if we put the restriction on $g(x)$ and $a(x)$ such that they are regular and monic polynomials, respectively, over $\mathbb{Z}_4$, then the minimal spanning set of $C$ reduces to that of \cite[Theorem 3]{abualrub}. To prove this, we first prove the following lemma, which appears as an exercise (Exercise XIII.6) in  \cite[p. 273]{mcdonald}.

\begin{lemma}Let $f(x)$ and $g(x)$ be two polynomials in $R[x]$. If $g(x)$ is regular, then there exists polynomials $q(x)$ and $r(x)$ such that $f(x)=g(x)q(x)+r(x)$, deg $r(x) < $ deg $g(x)$.
\end{lemma}

\begin{proof} \label{division algorithm} Since $g(x)$ is regular, by \cite[Theorem XIII.6]{mcdonald} there exists a monic polynomial $g^{*}(x) \in R[x]$ such that $g(x)=v(x) g^{*}(x)$, where $v(x)$ is a unit in $R[x]$.

Since $g^{*}(x)$ is monic, by division algorithm, there exists $q^\prime(x)$ and $r(x)$ in $R[x]$ such that $f(x)=g^{*}(x)q^\prime(x)+r(x)$, where deg $r(x) < $ deg $g^{*}(x)$. On multiplying both sides by $v(x)$, we get $v(x)f(x)=v(x)g^{*}(x)q^\prime(x)+v(x)r(x)$, from which we get $f(x)=g(x)q(x)+r(x)$, where $q(x)=(v(x))^{-1}q^\prime(x)$.

Since $g^{*}(x)$ is monic, so deg $g(x)$ $\geq$ deg $g^{*}(x)$, as deg $g(x)=$ deg $v(x)+$ deg $g^{*}(x)$. From this follows that deg $r(x) <$ deg $g(x)$.
 \end{proof}

The following result is a generalization of \cite[Theorem 3]{abualrub} in the present setting.

\begin{theorem}  Let $C=\left\langle g(x)+up(x), ~ua(x) \right\rangle$ be a cyclic code of length $n$ over $R$, and $g(x)$ is regular and $a(x)$ is monic in $\mathbb{Z}_4[x]$ with $deg~g(x)=k_1$ and $deg~a(x)=k_2$, respectively. Then $C$ has rank $n-k_2$ and a minimal spanning set $B=\{ (g(x)+up(x)),$ $x(g(x)+up(x)),$ $x^2(g(x)+up(x)), \cdots, $ $x^{n-k_1-1}(g(x)+up(x)),$ $~ua(x),$ $xua(x), x^2ua(x),$ $\cdots, x^{k_1-k_2-1}ua(x) \}$.
\end{theorem}

\begin{proof}
Suppose $C=\left\langle g(x)+up(x), ua(x) \right\rangle$ with $deg ~g(x)=k_1$ and $deg ~a(x)=k_2$, where $g(x)$ is regular and $a(x)$ is monic in $\mathbb{Z}_4[x]$. To prove $B$ is the minimal spanning set of $C$, it suffices to show that
$B$ spans the span of $A=\{ (g(x)+up(x)), x(g(x)+up(x)), x^2(g(x)+up(x)),$ $\cdots, x^{n-k_1-1}(g(x)+up(x)),$ $~ua(x), xua(x),x^2ua(x),\cdots,$  $x^{n-k_2-1}ua(x) \}$. For this, we first show that $x^{k_1-k_2}ua(x) \in $ Span $B$.

Since $g(x)$ is regular, then so is $(g(x)+up(x))$. By Lemma (\ref{division algorithm}), $x^{k_1-k_2}ua(x)=u(g(x)+up(x))q(x)+ur(x)$, where $r(x)=0$ or $\textrm{deg}~r(x) < k_1$, and $q(x) \in \mathbb{Z}_4[x]$. This implies that $ur(x) \in C$.
Since deg $r(x) <$ deg $(g(x)+up(x))$, so if $r(x) \ne 0$, then it cannot be expressed as a linear combination of $(g(x)+up(x))$ and its multiples. Therefore, $ur(x)=ua(x)b(x)$ for some $ b(x) \in R[x]$.

Since $a(x)$ is monic, so deg $ur(x) =$ deg $ua(x)+$ deg $b(x)$. From this follows that deg $b(x) \leq k_1 - k_2-1$. Thus, we get $x^{k_1-k_2}ua(x)=u(g(x)+up(x))q(x)+ua(x)b(x)$ with deg $b(x) \leq k_1-k_2-1$. It follows that $x^{k_1-k_2}ua(x) \in $ span $B$.

Similarly, we can show that $x^{k_1-k_2+1}ua(x)$, $x^{k_1-k_2+2} ua(x)$, $\cdots$, $x^{n-k_2-1}ua(x)$ are in span $B$. Hence $B$ is a minimal generating set of $C$.

To prove the linear independence of $B$, assume that
$s(x)(g(x)+up(x))=0$ $(\mbox{mod}~x^n-1)$  and $ut(x)a(x)=0~~(\mbox{mod}~x^n-1)$ for some $s(x)=s_0+s_1x+s_2x^2+\cdots+s_{n-k_1-1}x^{n-k_1-1}  \in R[x]$ and $t(x)=t_0+t_1x+t_2x^2+\cdots+t_{k_1-k_2-1}x^{k_1-k_2-1} \in \mathbb{Z}_4[x]$.

Since $g(x)+up(x)$ is regular, by \cite[Theorem XIII.6]{mcdonald} there exists a monic polynomial $g^{*}(x) \in R[x]$ such that $(g(x)+up(x))=v(x) g^{*}(x)$, where $v(x)$ is a unit $R[x]$.
Therefore, $s(x)v(x)g^{*}(x)=0~~(\mbox{mod}~x^n-1)$, which implies $s(x)g^{*}(x)=0~~(\mbox{mod}~x^n-1)$, as $v(x)$ is a unit in $R[x]$.
Let $g^{*}(x)=g^{*}_0+g^{*}_1x+g^{*}_2x^2+\cdots+g^{*}_{t}x^{t}$, where $t \leq n-k_1-1$. Then \[ \left(s_0+s_1x+s_2x^2+\cdots+s_{n-k_1-1}x^{n-k_1-1}\right)\left(  g^{*}_0+g^{*}_1x+g^{*}_2x^2+\cdots+g^{*}_{t}x^{t} \right) =0~~ (\mbox{mod}~x^n-1).\]
By comparing the coefficient of highest power of $x$ on both sides, we get $s_{n-k_1-1}g^{*}_t=0$, from which follows that $s_{n-k_1-1}=0$, as $g^{*}_t$ is a unit in $R$.
Again by comparing the coefficient of next highest power of $x$, we get $s_{n-k_1-1}g^{*}_{t-1}+s_{n-k_1-2}g^{*}_{t}=0$, which implies that $s_{n-k_1-2}=0$. On continuing this way, we get $s_i=0$ for $i=0,1,\ldots,s_{n-k_1-3}$. Similarly we can show that $t_i=0$ for all $i=0,1,\ldots,t_{k_1-k_2-1}$. Therefore $B$ is linearly independent.
 \end{proof}

\begin{theorem} Let $C= \left\langle g(x)+up(x), ua(x) \right\rangle$ be a cyclic code of length $n$ over $R$. Then $w_H(C)=w_H(\mbox{ker} ~\phi)$, i.e., $w_H(C)=w_H(\left\langle ua(x) \right\rangle)$.
\end{theorem}

\begin{proof} Let $c(x)=c_0(x)+uc_1(x) \in C$. Then $uc(x)=uc_0(x)$. It is clear that $w_H(uc(x))=w_H(uc_0(x)) \leq w_H(c(x))$. So $w_H(uC) \leq w_H(C)$. Also, since $uC$ is a subcode of $C$, $w_H(C) \leq w_H(uC)$. Hence the result.
 \end{proof}

\section{One generator cyclic codes over $R$}

We now consider cyclic codes over $R$ which are principal ideals in $\frac{R[x]}{\left\langle x^n-1\right\rangle}$. For a finite chain ring $\mathcal{R}$, the ring $\frac{\mathcal{R}[x]}{\left\langle x^n-1\right\rangle}$ is a principal ideal ring and the form of the generator of an ideal of $\frac{\mathcal{R}[x]}{\left\langle x^n-1\right\rangle}$ is well known \cite{norton}. Using the form of this generator, a necessary and sufficient condition for cyclic codes over $\mathbb{Z}_q$ to be free is provided in \cite[Proposition 1]{bhaintwal}. However, in the present case, $R$ is not a chain ring and the form of the generator of a principally generated ideal of $\frac{R[x]}{\left\langle x^n-1\right\rangle}$ is not known. Below we generalize \cite[Proposition 1]{bhaintwal} for the present case and provide a necessary condition (Theorem (\ref{thm2.4})) and a sufficient condition (Theorem (\ref{thm2.5})) for the cyclic codes over $R$ to be free.

\begin{theorem}\label{thm2.4}
Let $C$ be a principally generated cyclic code of length $n$ over $R$ generated by $g(x) \in R[x]$. If $g(x) \mid x^n-1$, then $C$ is $R$-free.
\end{theorem}

\begin{proof}
Suppose that $g(x) \mid x^n - 1$ and $x^n-1 = g(x)h(x)$. Since $x^n-1$ is a regular polynomial, $g(x)$ and $h(x)$ must also be regular polynomials. By \cite[Theorem XIII.6] {mcdonald}, there exist monic polynomials $g^\prime(x), h^\prime(x)$ such that $g(x) = v_1(x)g^\prime(x)$ and $h(x) = v_2(x)h^\prime(x)$ and $\overline{g}(x) = \overline{g^\prime}(x)$ and $\overline{h}(x) = \overline{h^\prime}(x)$,  where $v_1(x), v_2(x)\in R[x]$ are units. Therefore, $x^n-1 = g(x)h(x) = v_1(x)v_2(x)g^\prime(x)h^\prime(x)$. Since $x^n-1, g^\prime(x)$ and $h^\prime(x)$ are all monic, we must have $v_1(x)v_2(x) = 1$ and $x^n-1 = g^\prime(x)h^\prime(x)$. Let deg $g^\prime(x) = n-k$. Then deg $h^\prime(x) = k$.  We have $C = \left\langle g(x) \right\rangle = \left\langle v_1(x)g^\prime(x) \right\rangle = \left\langle g^\prime(x) \right\rangle$, as $v_1(x)$ is a unit. Obviously the set $S=\{g^\prime(x), xg^\prime(x), \ldots, x^{k-1}g^\prime(x)\}$ spans $C$.

Now suppose $a(x)g^\prime(x) = 0~(\mbox{mod} x^n -1)$ for some $a(x) \in R[x]$ with deg $a(x) < k$. Then $x^n-1 \mid a(x)g^\prime(x)$, which implies that $\frac{x^n-1}{g^\prime(x)} \mid a(x)$, i. e., $h^\prime(x) \mid a(x)$. Since $h^\prime(x)$ is monic polynomial of degree $k$, it cannot divide a non-zero polynomial of degree less than $k$. It follows that $a(x)=0$. So the set $S$ is linearly independent and thus forms a basis for $C$. Hence $C$ is an $R$-free code.
 \end{proof}

We have following converse of Theorem (\ref{thm2.4}).

\begin{theorem} \label{thm2.5}
Let $C$ be a principally generated cyclic code of length $n$ over $R$ generated by $g(x) \in R[x]$. If $C$ is $R$-free, then there exists a monic generator $g^\prime(x)$ of $C$ such that $g^\prime(x) \mid x^n-1$.
\end{theorem}

\begin{proof}
Suppose that $C$ is an $R$-free code. Since $g(x)$ generates an $R$-free code, $g(x)$ must be a regular polynomial. Therefore there exist a monic polynomial $g^\prime(x) \in R[x]$ such that $g(x) = v(x)g^\prime(x)$ and $\overline{g}(x) = \overline{g^\prime}(x)$,  where $v(x)$ is a unit in $R[x]$. Let the $R$-rank of $C$ be $s$ and $S=\{c_1, c_2, \ldots, c_s\}$ an $R$-basis of $C$. Then the set $\{\overline{c}_1, \overline{c}_2, \ldots, \overline{c}_s\}$ forms a basis for the cyclic code $\overline{C}$ over the finite field $\overline{R}$. Since $C = \left\langle g(x) \right\rangle$, so $\overline{C} = \left\langle \overline{g}(x) \right\rangle = \overline{g^\prime}(x)$. Since $\overline{g^\prime}(x)$ is monic, therefore it is the generator polynomial of $\overline{C}$.  Let deg $\overline{g^\prime}(x) = n-k$. Then the set $\{\overline{g^\prime}(x), x\overline{g^\prime}(x), \ldots, x^{k-1}\overline{g^\prime}(x)\}$ forms a basis for $\overline{C}$. So we must have $s=k$.

Now $C = \left\langle g(x) \right\rangle = \left\langle g^\prime(x) \right\rangle$. Clearly, the elements $g^\prime(x), xg^\prime(x), x^2g^\prime(x) \ldots $ span $C$. Also, the elements $\{g^\prime(x), xg^\prime(x), \ldots, x^{k-1}g^\prime(x)\}$ are linearly independent over $R$; for if they are not, then they give a dependence relation among the elements $\overline{g^\prime}(x), x\overline{g^\prime}(x), \ldots, x^{k-1}\overline{g^\prime}(x)$, a contradiction. Now since $x^kg^\prime(x)$ is a codeword, we can write $x^kg^\prime(x)$ as a linear combination of the elements $x^ig^\prime(x), i=0, 1, \ldots, k-1$. Let
\[ x^kg^\prime(x) = \sum_{i=0}^{k-1} a_ix^ig^\prime(x)~,\]
which can be written as $\sum_{i=0}^{k} a_ix^ig^\prime(x) =0$ with $a_k=1$, or $a(x)g^\prime(x)=0$. Then $x^n-1 \mid a(x)g^\prime(x)$ and since $a(x)g^\prime(x)$ is a monic polynomial of degree $n$,  we must have $x^n-1=a(x)g^\prime(x)$. Therefore, $g^\prime(x) \mid x^n-1$.
 \end{proof}

The following result follows from Theorem (\ref{thm2.4}) and Theorem (\ref{thm2.5}).

\begin{proposition} \label{R free to z4 free} Let $C$ be a principally generated cyclic code of length over $R$. Then $C$ is free if and only if there exists a monic generator $g(x)$ in $C$ such that $g(x) ~| ~x^n-1$. Furthermore, $C$ has free rank $n- deg~g(x)$ and the elements $g(x)$, $xg(x)$, $\cdots$, $x^{n-deg~g(x)-1}g(x)$ forms a basis for $C$.
\end{proposition}

\begin{example}
Consider the cyclic code $C$ of length $7$ over $R$ generated by the polynomial $g(x) = x^3+2x^2+x-1$. $g(x)$ is the Hensel lift of $x^3+x+1 \in \mathbb{F}_2[x]$ to $R$. The cyclic code  $C = \left\langle g(x) \right\rangle$ an $R$-free cyclic code of length $7$ and the free rank $4$.
\end{example}

\begin{theorem} If $C=C_1+uC_2$ is free cyclic code over $R$ then so is $C_1$ over $\mathbb{Z}_4$.
\end{theorem}
\begin{proof} From Proposition (\ref{R free to z4 free}), if $C$ is a free cyclic code over $R$ with generator polynomial $g(x)$ then $x^n-1=g(x) h(x)$. Since $R=\mathbb{Z}_4+u\mathbb{Z}_4$, we can express $g(x)=g^{'}(x)+ug^{''}(x)$ and $h(x)=h^{'}(x)+uh^{''}(x)$, where $g^{'}(x), g^{''}(x), h^{'}, h^{''}(x) \in \mathbb{Z}_4[x]$. Then $x^n-1=g^{'}(x) h^{'}(x)~(\mbox{mod}~ u)$. The result follows.
 \end{proof}

\begin{example}\label{ex1}
 Consider again the cyclic code $C$ of length $7$ generated by $g(x)=x^3+2x^2+x-1$. Then $C$ is free over $R$ since $x^3+2x^2+x-1$ is divisors of $x^7-1$ over $R$.  As $x^3+2x^2+x-1$ is divisors of $x^7-1$ over $\mathbb{Z}_4$ as well, $C_1$ is a free cyclic code of length $7$ over $\mathbb{Z}_4$.
\end{example}

A polynomial $e(x)$ in $R[x]$ is said to be an \emph{idempotent} if $e(x)^2=e(x)~(\mbox{mod}~ x^n-1)$. The following theorems are the generalization of \cite[Theorem 5, 6]{pless}.

\begin{theorem} Let $C$ be a cyclic code of length $n$ over $R$.
\begin{enumerate}
\item If $C= \left\langle g \right\rangle$ and $g | x^n-1$, then $C$ has an idempotent generator in $R$.
\item If $C= \left\langle ug \right\rangle$ with $g | x^n-1$, then $C=\left\langle ue \right\rangle$, where $e$ is an idempotent generator of $C$. 
\end{enumerate}
\end{theorem}
\begin{proof} Let $x^n-1=gh$ for some $h$ in $R[x]$. Since $x^n-1$ has distinct factors over $R$, therefore $g,h$ are coprime in $R[x]$. Then there exist $\lambda_1, \lambda_2$ in $R[x]$ such that $g \lambda_1 + h \lambda_2=1$.

Let $e=g \lambda_1$. Then $e \in \left\langle g \right\rangle$. Since $g \lambda_1 + h \lambda_2=1$, $e=1- h \lambda_2$ and $e^2=e (1-h \lambda_2)=e ~(\mbox{mod}~ x^n-1)$. Now $ge=g(1-h \lambda_2)=g~(\mbox{mod}~ x^n-1)=g$. This implies that $g \in \left\langle e \right\rangle$. Hence $\left\langle e \right\rangle = \left\langle g \right\rangle$.

 The second result can be proved similarly.
 \end{proof}

\begin{theorem}
If $C$ be a free cyclic code of length $n$ over $R$ with idempotent generator $e(x)$ in $R[x]$ then $C^{\perp}$ has the idempotent $1-e(x^{-1})$.
\end{theorem}
\begin{proof} Similar to the finite fields case.  \end{proof}

\subsection{One generator cyclic codes as $n^{th}$ roots of unity}

Since $(n,4)=1$, so $x^n-1$ factorizes uniquely into coprime monic basic irreducible polynomials. From Theorem (\ref{primitive}), there exists a primitive $n^{th}$ root of unity in $GR(R,r)$. Let $\xi^{i_1}, \xi^{i_2}, \cdots, \xi^{i_k}$ be $n^{th}$ roots of unity in $GR(R,r)$. Define the minimal polynomial $M_i(x)$ of $\xi^{i}$ as the monic polynomial of least degree having a root $\xi^i$ over $R$.   Then a cyclic code $C$ of length $n$ over $R$ can also be described in terms of $n^{th}$ roots of unity.  Then the cyclic code $C$ can be defined as \[  C= \{ c(x) \in R_n ~~: ~~c(\xi^{i_j})=0,~1 \leq j \leq k \}. \]
The generator polynomial $g(x)$ of $C$ is the least common multiple of minimal polynomials of $\xi^{i_j}$, $1 \leq j \leq k$. Then $g(x)~ |~ (x^n-1)$. Hence $C$ is a free code over $R$.

The following is a straightforward generalization of \cite[Proposition 2]{bhaintwal}.

\begin{proposition} \label{nthroot} \cite{bhaintwal} Suppose that the generator polynomial $g(x)$ of a cyclic code $C$ of length $n$ over $R$ divides $(x^n-1)$ and has as roots $\xi^b, \xi^{b+1}, \cdots \xi^{b+\delta-1}$, where $\xi$ is a primitive $n^{th}$ root of unity in a Galois extension of $R$. Then $d(C) \geq \delta$.
\end{proposition}

\begin{example}
Let $\xi$ be a root of the basic primitive polynomial $f(x)=x^4+3x^3+2x^2+1$, which is a factor of $x^{15}-1$ over $R$. 
Let the generator polynomial of a cyclic code of length $15$ over $R$ is defined as $g(x)=lcm(M_0(x),M_1(x),M_2(x),M_3(x),M_4(x),$ $M_5(x),M_6(x))$, where $M_i(x)$ are the minimal polynomials of $\xi^i$, $i=0,1,2,3,4,5,6$, respectively. We have $M_0(x)=x-1$, $M_1(x)=M_2(x)=M_4(x)=x^4+3x^3+2x^2+1$, $M_3(x)=M_6(x)=x^4+x^3+x^2+x+1$ and $M_5(x)=x^2+x+1$. Therefore, $g(x)=x^{11}+2x^9+3x^8+3x^7+x^6+2x^4+3x^3+x^2+3x+3$. The cyclic code $C$ generated by $g(x)$ is a free code of rank $4$. Since $g(x)$ has $7$ consecutive roots, $d(C) \ge 8$, where $d(C)$ denotes the minimum Hamming distance of $C$. Also since $2g(x)=8$, we must have $d(C) = 8$.
\end{example}

\section{Conclusion}

In this paper we have studied some structural properties of cyclic codes of odd length over the ring $R= \mathbb{Z}_4+u\mathbb{Z}_4$, $u^2=0$. The general form of the generators of cyclic codes over $R$ is provided and a formula for their ranks is determined. We have mainly focused on cyclic codes over $R$ that are principally generated. We have also obtained a necessary condition and a sufficient condition for such codes to be free $R$-modules.

\medskip
Received xxxx 20xx; revised xxxx 20xx.
\medskip

\end{document}